\documentclass[11pt, a4paper]{article}
\usepackage{fullpage}
\usepackage{amsmath,amssymb,amsthm,mathtools,cite,comment,graphicx}
\usepackage{hyperref}

\hypersetup{hidelinks}

\usepackage{tikz}
\usetikzlibrary{arrows.meta, positioning,calc}

\theoremstyle{plain}
\newtheorem{theorem}{Theorem}[section]
\newtheorem{lemma}[theorem]{Lemma}

\theoremstyle{definition}

\theoremstyle{remark}
\newtheorem{remark}[theorem]{Remark}

\newcommand{\problemdef}[3]{
    \begin{center}
    \fbox {   \parbox[c]{0.95\textwidth}{
        \textsc{\large #1} 
        
         \textbf{Input:} #2 \\
         \textbf{Goal:} #3 
        }}
    \end{center} 
}

\newcommand{\cP}{\mathcal{P}}
\newcommand{\cQ}{\mathcal{Q}}
\newcommand{\cS}{\mathcal{S}}

\title{A Linear-Time Algorithm\\for Finding an Odd/Even Cycle Through Two Specified Vertices}
\author{
Takumi Kano\thanks{Namiki Secondary School, Ibaraki, Japan.}
\and
Yutaro Yamaguchi\thanks{Osaka University, Osaka, Japan. Email: \texttt{yutaro.yamaguchi@ist.osaka-u.ac.jp}}
}
\date{\empty}

\graphicspath{{figs/}}

\begin{document}

\maketitle
\thispagestyle{empty}

\begin{abstract}
We present a deterministic linear-time algorithm for finding an odd/even cycle through two specified vertices in an undirected graph.
This is shown in a generalized form as follows:
Let $\Gamma$ be any group in which every element is of order at most $2$.
For a given $\Gamma$-labeled graph with two specified vertices (or edges), we can determine in linear time whether there exist two cycles with distinct labels that are through both of the two specified vertices (or edges), and find such cycles if yes.

\medskip

\end{abstract}

\newpage
\pagenumbering{roman}
\tableofcontents
\newpage
\pagenumbering{arabic}
\setcounter{page}{1}

\section{Introduction}
In this paper, we show the following:

\begin{theorem}\label{thm:odd}
    There exists a deterministic algorithm to test for a given undirected graph of $n$ vertices and $m$ edges with two specified vertices $s$ and $t$, whether a cycle of each parity (odd/even length) through both $s$ and $t$ exists or not, and if exists, to find such a cycle in $\mathrm{O}(n + m)$ time.
\end{theorem}

This extends a linear-time algorithm for finding a cycle through three specified vertices in a given undirected graph \cite{fleischner1992detecting}.
To see this, let $G$ be an undirected graph, and $x, y, z$ be three distinct vertices in $G$.
We construct from $G$ an auxiliary graph $\tilde{G}$ as follows:
subdivide each edge in $G$ with exactly one extra vertex in the middle, split the vertex $z$ into two copies $z_1$ and $z_2$ with the same neighbors as the original $z$, and add an extra edge between $z_1$ and $z_2$.
Then, a cycle in $\tilde{G}$ is odd if and only if it is through the extra edge $\{z_1, z_2\}$.
Thus, an odd cycle through both $x$ and $y$ in $\tilde{G}$ corresponds to a cycle through all of $x$, $y$, and $z$ in $G$.

In order to prove Theorem~\ref{thm:odd}, we actually solve a more general problem in group-labeled graphs, which is formally defined in Section~\ref{sec:group-labeled}.
Our main theorem is stated as follows:

\begin{theorem}\label{thm:main}
    Let $\Gamma$ be a group in which every element is of order at most $2$ (i.e., self-inverse).
    There exists a deterministic algorithm to determine for a given biconnected $\Gamma$-labeled graph of $n$ vertices and $m$ edges with two specified vertices $s$ and $t$ (or two specified edges $e_s$ and $e_t$), whether the label of a cycle through both $s$ and $t$ (or $e_s$ and $e_t$) is unique or not in $\mathrm{O}(n + m)$ time.
    Moreover, if it is not unique, the algorithm outputs two such cycles with distinct labels.
\end{theorem}

The assumption that every element of $\Gamma$ is of order at most $2$ is essential.
For signed graphs, that is, for $\Gamma=\mathbb{Z}_2 = \mathbb{Z} / 2\mathbb{Z}$, DeVos and Nurse~\cite{devos2023cycles} recently gave a structural characterization for triconnected graphs to have two cycles $C_1$ and $C_2$ passing through two specified edges such that the signs of $C_1$ and $C_2$ are different.
This characterization (rephrased as follows) plays a key role in our algorithm, for which we give a constructive proof with an efficient algorithm. 

\begin{theorem}[DeVos and Nurse \cite{devos2023cycles} (rephrased)]\label{thm:DN2023}
    Let $G$ be a triconnected $\mathbb{Z}_2$-labeled graph, and $e_1, e_2 \in E(G)$ be disjoint edges that have no parallel edges.
    Then, $G$ contains two cycles $C_1$ and $C_2$ passing through both $e_1$ and $e_2$ such that the labels of $C_1$ and $C_2$ are different if and only if $G - \{e_1, e_2\}$ contains a non-zero cycle of length at least three.
\end{theorem}

\begin{remark}
Indeed, this characterization fails already for triconnected $\mathbb{Z}_3$-labeled graphs, as shown in Figure~\ref{fig:mod3-counterexample}.
While each edge with non-zero label has an orientation, it can be traversed in both directions and the label is negated if it is traversed in the backward direction; for more precise definitions on group-labeled graphs, see Section~\ref{sec:group-labeled}.
Even though $G-\{e_1,e_2\}$ contains a cycle on $\{x_1, x_2, w, v, u\}$ with label $2 - 1 + 0 + 0 + 0 = 1$, which is non-zero, it may still happen that every cycle through the distinguished edges $e_1$ and $e_2$ has the same label $0$.
\end{remark}

\begin{figure}[t]
    \centering
    \scalebox{0.9}{
        \def\svgwidth{0.6\columnwidth}
        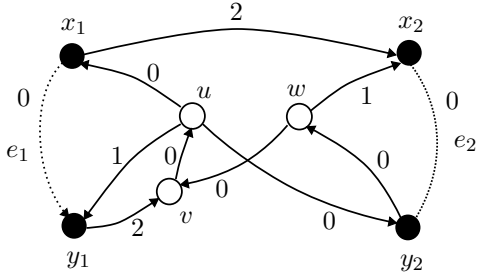
    }

    \caption{A triconnected counterexample for Theorem~\ref{thm:DN2023} when $\Gamma=\mathbb{Z}_3$. The solid edges are directed and labeled by elements of $\mathbb{Z}_3$; the label of a walk is the signed sum of the traversed edge labels, taking sign $+1$ when the walk follows the chosen direction and sign $-1$ otherwise (see Section~\ref{sec:group-labeled} for more details).}
    \label{fig:mod3-counterexample}
\end{figure}

The proof of Theorem~\ref{thm:DN2023} in \cite{devos2023cycles} is by a minimal-counterexample
argument with an extremal choice of a non-zero cycle.
Thus, it does not provide a fast algorithm to find explicit witness cycles even for $\mathbb{Z}_2$-labeled graphs.
Our main contribution is showing Theorem~\ref{thm:main}, which gives a linear-time algorithm to find target cycles for a more general class of group-labeled graphs.

We should remark that a polynomial-time algorithm for our problem (when $|\Gamma|$ is constant) can be derived from \cite{huynh2009linkage,kawarabayashi2011graph}.
In \cite{kawarabayashi2011graph}, Kawarabayashi, Reed, and Wollan gave an $\mathrm{O}(mn\cdot\alpha(m, n))$-time algorithm for finding $k$ vertex-disjoint paths in a given undirected graph between $k$ specified pairs of vertices with a specified parity of each path for any fixed $k$.
It is not difficult to see that, by taking $k = 2$ and checking all of $2^2 = 4$ patterns of parities for each element of a fixed generating set of $\Gamma$, we can solve our problem via this setting.
In \cite{huynh2009linkage}, Huynh also provided a polynomial-time algorithm for finding $k$ vertex-disjoint paths in a given $\Gamma$-labeled graph between $k$ specified pairs of vertices with a specified label for any fixed $k$ and a constant-size group $\Gamma$.
A similar reduction works for this setting.

Nevertheless, their algorithms are built on sophisticated graph minor theory initiated by Robertson and Seymour (in particular, \cite{robertson1995graph}) and do not achieve linear-time computation.
The advantage of our algorithm is that it is rather elementary and runs in linear time.

The remainder of this paper is organized as follows.
Section~\ref{sec:preliminaries} describes necessary definitions, problems, and useful basic facts.
Section~\ref{sec:isolate} isolates a single nontrivial subproblem in triconnected graphs from our problem with the aid of the so-called SPQR-trees, for which we state the linear-time solvability as Theorem~\ref{triconnected-problem}.
Section~\ref{sec:triconnected} proves Theorem~\ref{triconnected-problem}, which completes our linear-time algorithm.

\section{Preliminaries}\label{sec:preliminaries}
Throughout the paper, graphs may have parallel edges, but have no self-loops.

\subsection{Graphs}
Let $G = (V, E)$ be an undirected graph.
An alternating sequence $W = (v_0, e_1, v_1, \dots, e_\ell, v_\ell)$ of vertices and edges in $G$ is called a \emph{walk} in $G$ if $e_i = \{v_{i-1}, v_i\} \in E$ for every $i = 1, 2, \dots, \ell$.
We call $\ell$ the \emph{length} of $W$, and we say that $W$ is \emph{odd} and \emph{even} if $\ell$ is odd and even, respectively.
We say that $W$ is a \emph{path} if all its vertices $v_0, v_1, \dots, v_\ell$ are distinct, and a \emph{cycle} if $\ell \ge 2$, $v_1, v_2, \dots, v_\ell$ are distinct, $v_0 = v_\ell$, and $e_1 \neq e_\ell$.
In particular, a walk of length zero is regarded as a path, which is called \emph{trivial}, and a walk consisting of two distinct parallel edges is regarded as a cycle.
A path is referred to as a $v_0$--$v_\ell$ path (or an $X$--$Y$ path when $v_0 \in X$ and $v_\ell \in Y$) if the end vertices are emphasized.
We denote by $P[v_i, v_j]$ the subpath of a path $P$ from $v_i$ to $v_j$, where the direction is flipped accordingly.

We say that $G$ is \emph{connected} if for every pair of vertices, say $s, t \in V$, there exists an $s$--$t$ path in $G$.
For each integer $k \ge 1$, a \emph{$k$-cut} is a vertex set $X \subseteq V$ with $|X| \le k$ such that its removal results in a disconnected graph $G - X$.
For $k \ge 2$, we say that $G$ is \emph{$k$-connected} if $|V| > k$ and it has no $(k-1)$-cut.
In particular, it is also called \emph{biconnected} when $k = 2$ and \emph{triconnected} when $k = 3$.
A \emph{bridge} is an edge whose removal makes the graph disconnected, and a \emph{block} is a maximal biconnected subgraph or a two-vertex subgraph with at least two parallel edges which becomes a bridge if we merge them into a single edge.
It is well-known that any connected graph has a unique tree structure formed by the blocks and bridges that are connected by the $1$-cuts, called the \emph{block-cut tree}, which can be computed in linear time~\cite{hopcroft1973algorithm}.
Any cycle is included in a single block as it cannot use a bridge.
For any pair of vertices $s$ and $t$ in the same block, there exists a cycle through $s$ and $t$ (by Menger's theorem~\cite{Menger1927}); conversely, such a cycle must be included in the unique block that contains both $s$ and $t$.
A further decomposition of biconnected graphs into triconnected components is also known \cite{hopcroft1973dividing}, and we explain a closely related concept, the so-called \emph{SPQR-trees}, in Section~\ref{sec:SPQR-tree}.

\subsection{Group-Labeled Graphs}\label{sec:group-labeled}
We first introduce a general definition of group-labeled graphs.
Let $\Gamma$ be a group, for which we use the multiplicative notation without assuming commutativity at first.
A \emph{$\Gamma$-labeled graph} is a pair $(\vec{G},\psi)$ of a directed graph $\vec{G}=(V,\vec{E})$ and a label function $\psi\colon\vec{E}\to\Gamma$.
Let $G=(V,E)$ denote the underlying graph of $\vec{G}$, i.e., $E\coloneqq\{\, e=\{u,v\} \mid \vec{e}=uv\in \vec{E}\,\}$, and define $\psi(e,uv)\coloneqq\psi(\vec{e})$ and $\psi(e,vu)\coloneqq \psi(\vec{e})^{-1}$ for each edge $e=\{u,v\}\in E$ with the corresponding arc $\vec{e}=uv\in \vec{E}$.
The \emph{label} of a walk $W = (v_0, e_1, v_1, \dots, e_\ell, v_\ell)$ is defined as
\[\psi(W) \coloneqq \psi(e_1, v_0v_1) \cdot \psi(e_2, v_1v_2) \cdot \dots \cdot \psi(e_\ell, v_{\ell-1}v_\ell).\]
A walk is called \emph{non-zero} if its label is not the identity element of $\Gamma$.
We say that $(\vec{G}, \psi)$ is \emph{balanced} if it contains no non-zero cycle.

In this paper, we focus on the case where every element of $\Gamma$ is of order at most $2$; that is, for any $\alpha \in \Gamma$, we have $\alpha^{-1} = \alpha$.
Then, $\psi(e, uv) = \psi(e, vu)$ for any edge $e = \{u, v\} \in E$, so the orientation of the edges becomes irrelevant.
Thus, we identify the original directed graph $\vec{G}$ with its underlying graph $G$, and each arc $\vec{e} = uv \in \vec{E}$ with its corresponding edge $e = \{u, v\} \in E$.
Also, in this case, $\Gamma$ is automatically abelian.
Indeed, for any $\alpha, \beta \in \Gamma$, we have
$\alpha\beta = (\alpha\beta)^{-1} = \beta^{-1}\alpha^{-1} = \beta\alpha$.
Thus, we use the additive notation hereafter.

Regarding the computation on $\Gamma$, we assume the black-box model (cf.~\cite{iwata2022finding,kawase2020finding}) in which each of the following elementary operations can be performed in constant time:
for each $\alpha, \beta \in \Gamma$, test whether $\alpha = \beta$ and compute $\alpha + \beta$.

Our problems are formally stated as follows (recall that any cycle through $s$ and $t$ is included in the unique block that contains $s$ and $t$, which enables us to assume biconnectivity, and then there exists at least one such cycle):

\problemdef{Two Labels of Cycles Through Two Vertices}
{A biconnected $\Gamma$-labeled graph $(G = (V, E), \psi)$ and two distinct vertices $s,t\in V$.}
{Determine whether there exist two cycles $C_1$ and $C_2$ in $G$ that intersect both $s$ and $t$ with $\psi(C_1) \neq \psi(C_2)$ (or all such cycles have the same label), and if so, find such cycles.}

\problemdef{Two Labels of Cycles Through Two Disjoint Edges}
{A biconnected $\Gamma$-labeled graph $(G=(V,E),\psi)$ and two disjoint edges $e_s,e_t\in E$.}
{Determine whether there exist two cycles $C_1$ and $C_2$ in $G$ that traverse both $e_s$ and $e_t$ with $\psi(C_1)\neq\psi(C_2)$ (or all such cycles have the same label), and if so, find such cycles.}

Theorem~\ref{thm:main} states that these problems admit a deterministic linear-time algorithm (note that when $e_s$ and $e_t$ share an end vertex, the problem is easily solved by Lemma~\ref{two-labels} below; see Lemma~\ref{lem:shared-endpoint} for the complete proof).
If $\Gamma = \mathbb{Z}_2 = \mathbb{Z} / 2\mathbb{Z}$ and $\psi(e) = 1$ for every $e \in E$, the former problem exactly requires to determine the possible parities of cycles in $G$ through both $s$ and $t$ and to find such a cycle of each possible parity.
Thus, Theorem~\ref{thm:odd} holds.

We also note that the two problems are essentially equivalent, i.e., each easily reduces to the other in linear time.
Indeed, from the latter to the former, it is enough to subdivide each of the edges $e_s$ and $e_t$ with an extra vertex $s$ and $t$, respectively, and to assign the original label to one of the two subdivided edges and label $0$ to the other.
This operation preserves biconnectivity and the labels of target cycles.

From the former to the latter, we first assume that $s$ and $t$ are not adjacent by subdividing each edge between $s$ and $t$ (with assigning the same label to one of the resulting edge and label $0$ to the other) if any; this neither changes the problem nor violates the biconnectivity.
Then, split $s$ into two vertices $s_1$ and $s_2$, add an edge $e_s = s_1s_2$ with label $0$, and replace each edge $su$ by two edges $s_1u$ and $s_2u$, both with label $\psi(su)$; do the same at $t$, obtaining $t_1,t_2$, and a new edge $e_t$ with label $0$.
The resulting graph is still biconnected, and contracting $e_s$ and $e_t$ (and merging parallel edges of the same label into a single one) recovers the original graph, so the cycles through $s$ and $t$ correspond exactly to the cycles through $e_s$ and $e_t$, and the labels are preserved.

Thus, in what follows, we concentrate on the latter problem.

\subsection{Basic Facts}
We first observe a shortcut of a path.

\begin{lemma}\label{first-hit}
Let $G=(V,E)$ be a graph, $s \in V$, $T\subseteq V$, and $Q$ be a path from $s$ to a vertex of $T$.
Then, $Q$ has a subpath $Q'$ starting at $s$, ending in $T$, and satisfying $|V(Q') \cap T| = 1$.
\end{lemma}

For $s \in V$ and $T \subseteq V \setminus \{s\}$, an $s$--$T$ $k$-fan is a collection of $k$ $s$--$T$ paths such that they only share the vertex $s$ and each path has no inner vertex in $T$.
By Menger's theorem~\cite{Menger1927}, any $k$-connected graph contains a collection of $k$ vertex-disjoint paths between two vertex subsets (as well as a $k$-fan), which can be computed in $\mathrm{O}(k(|V| + |E|))$ time using the Ford--Fulkerson algorithm for maximum flow problem~\cite{ford1956maximal} (see also \cite[Section~9.2]{schrijver2003combinatorial}).

\begin{lemma}\label{extended-menger2}
There exists a deterministic algorithm running in $\mathrm{O}(k(|V|+|E|))$ time that finds for a given $k$-connected graph $G = (V, E)$ and two vertex sets $A, B \subseteq V$ with $|A| \ge k$ (or $A = \{s\}$) and $|B| \ge k$, a collection of vertex-disjoint $A$--$B$ paths (or an $s$--$B$ $k$-fan, respectively) such that all the inner vertices are disjoint from $A \cup B$.
\end{lemma}

It is well-known that the following elementary problems on (general) group-labeled graphs can be solved in linear time (see \cite{kawase2020finding, iwata2022finding} for details in more general form):
\begin{itemize}
    \item[(E1)] test whether a group-labeled graph is balanced or not, and if not, find a non-zero cycle;
    \item[(E2)] find a non-zero cycle through a specified vertex $s$ (if any);
    \item[(E3)] find a non-zero $s$--$t$ path between two specified vertices $s$ and $t$ (if any).
\end{itemize}

Regarding this, we introduce an operation on group-labeled graphs.
Let $(G = (V, E), \psi)$ be a $\Gamma$-labeled graph.
We say that $\psi'$ is obtained from $\psi$ by \emph{shifting at $v \in V$ by $\alpha \in \Gamma$} if
\[\psi'(e) = \begin{cases}
    \psi(e) + \alpha & \text{if $e$ is incident to $v$},\\
    \psi(e) & \text{otherwise}.
\end{cases}\]
By definition, this operation does not change the label of any cycle or any path whose end vertices are not $v$, and changes the labels of all paths one of whose end vertices is $v$ uniformly by $\alpha$.
Two $\Gamma$-labeled graphs are called \emph{equivalent} if one is obtained from the other by a sequence of shifting operations.

The first problem (E1) can be solved based on the following observations (proved for general groups $\Gamma$ in \cite[Section~2.2.2]{yamaguchi2016combinatorial}).

\begin{lemma}[cf.~{\cite[Proposition~2.4]{yamaguchi2016combinatorial}}]\label{lem:tree-gauge-normalization}
    For any connected $\Gamma$-labeled graph $(G = (V, E), \psi)$ and any spanning tree $T$ of $G$, there exists an equivalent $\Gamma$-labeled graph $(G, \psi')$ such that $\psi'(e) = 0$ for every $e \in T$.
    Moreover, such $\psi'$ can be computed in $\mathrm{O}(|V| + |E|)$ time.
\end{lemma}

\begin{lemma}[cf.~{\cite[Proposition~2.5]{yamaguchi2016combinatorial}}]\label{lem:balanced}
    A $\Gamma$-labeled graph is balanced if and only if it is equivalent to a $\Gamma$-labeled graph whose label function is trivial, i.e., all the edges are assigned label $0$.
\end{lemma}

It is easy to see that the second problem (E2) can also be solved by Lemma~\ref{lem:balanced} combined with the block-cut tree and Menger's theorem (such a cycle exists if and only if the block that $s$ belongs to has a non-zero cycle $C$; if $C$ does not contain $s$, one can find an $s$--$V(C)$ $2$-fan in linear time, which splits $C$ as a chord into two cycles through $s$ and at least one is non-zero).
In a similar way, the third problem (E3) can be solved in a slightly more general form as follows:

\begin{lemma}[cf.~{\cite[Proposition~8]{kawase2020finding}}]\label{two-labels}
Let $(G=(V,E), \psi)$ be a $\Gamma$-labeled graph, $s,t\in V$, and $L$ be the set of labels of all $s$--$t$ paths in $G$. Then one can compute, in $\mathrm{O}(|V|+|E|)$ time, a set $X\subseteq L$ and, for each $\lambda\in X$, an explicit $s$--$t$ path $P_\lambda$ with $\psi(P_\lambda)=\lambda$, such that $|X|=\min\{|L|,2\}$.
\end{lemma}

\section{Reduction via SPQR-Tree}\label{sec:isolate}
In this section, we reduce our problem \textsc{Two Labels of Cycles Through Two Disjoint Edges} with the aid of the SPQR-trees to a unique nontrivial task, naming it \textsc{Two Labels of Cycles Through Two Edges in Triconnected Graphs}.

\subsection{SPQR-Trees}\label{sec:SPQR-tree}

An \emph{SPQR-tree} represents the decomposition of a biconnected graph into triconnected components, which was introduced by Di Battista and Tamassia \cite{di1989incremental}. We review its definition and properties according to \cite{gutwenger2000linear,gutwenger2010application}.

Let $G=(V,E)$ be a biconnected graph. The SPQR-tree is obtained by recursively decomposing $G$ along \emph{split pairs}, i.e., pairs of vertices that form either a $2$-cut (also called a \emph{separation pair}) or a pair of adjacent vertices, and by replacing each separated part by a \emph{virtual edge}. Thus, each decomposition step splits the graph into two parts sharing exactly the corresponding split pair, and the original graph can be recovered by expanding the virtual edges recursively.

Each node $\mu$ of the SPQR-tree has an associated graph, called the \emph{skeleton} of $\mu$ and denoted by $\mathrm{skel}(\mu)$. There are four types of nodes. A \emph{Q-node} corresponds to a single original edge of $G$, an \emph{S-node} has a cycle as its skeleton, a \emph{P-node} has a skeleton consisting of two vertices joined by parallel edges, and an \emph{R-node} has a triconnected skeleton. In particular, each original edge $e\in E$ corresponds to a unique Q-node, which we denote by $\mathrm{Q}(e)$. Moreover, the skeletons of the S-, P-, and R-nodes are exactly the components in the triconnected component decomposition of $G$.

Each skeleton edge is either \emph{real} or \emph{virtual}. A real edge is an original edge of $G$, while a virtual edge is newly introduced for the decomposition. If two nodes $\mu$ and $\nu$ are adjacent in the SPQR-tree, then $\mathrm{skel}(\mu)$ and $\mathrm{skel}(\nu)$ contain corresponding virtual edges with the same pair of end vertices. This common pair of end vertices is the split pair corresponding to the tree edge $\{\mu, \nu\}$. Thus, deleting the tree edge $\{\mu, \nu\}$ splits the tree into two components, and each side represents a subgraph attached to the other side through that split pair.

Once the tree is rooted (at a node or a connected node set), every non-root node $\nu$ has a unique parent $\mu$ that is adjacent to $\nu$ and closer to the root(s) than $\nu$.
The virtual edge in $\mathrm{skel}(\nu)$ corresponding to this tree edge $\{\nu, \mu\}$ is called the \emph{parent virtual edge} of $\nu$, whose end vertices are called the \emph{poles} of $\nu$. The \emph{pertinent graph} $\mathrm{pert}(\nu)$ is obtained from $\mathrm{skel}(\nu)$ by expanding every virtual edge except the parent one. Equivalently, if $e_\nu$ is the parent virtual edge of $\nu$, then $\mathrm{pert}(\nu)-e_\nu$ is exactly the subgraph of $G$ represented by the subtree rooted at $\nu$, and it is attached to the rest of the graph only through the poles of $\nu$.

The SPQR-tree is unique up to isomorphism, can be computed in linear time, and the total size of all skeletons is $\mathrm{O}(|V|+|E|)$. In the reduction below, we focus on the unique tree path between the Q-nodes corresponding to the distinguished edges, together with the off-path subtrees attached to it.

\subsection{Reduction}\label{sec:reduction}

Let $(G=(V,E),\psi)$ be a biconnected $\Gamma$-labeled graph, and let $e_s, e_t\in E$ be disjoint.

\medskip
\noindent
\textbf{Step 1.}~
Let $T$ be an SPQR-tree of $G$, and let $\mathrm{Q}(e_s)=\mu_0,\mu_1,\dots,\mu_k=\mathrm{Q}(e_t)$
be the nodes on the unique path in $T$ from $\mathrm{Q}(e_s)$ to $\mathrm{Q}(e_t)$.
For each $i=0,1,\dots,k$, let $\ell_i$ and $r_i$ denote the left and right interface edges of $\mu_i$, namely, $\ell_0 \coloneqq e_s$, $r_k \coloneqq e_t$,
and for $1\le i\le k$, let $\ell_i$ be the virtual edge of $\mathrm{skel}(\mu_i)$ corresponding to the tree edge $\{\mu_{i-1}, \mu_i\}$, while, for $0\le i\le k-1$, let $r_i$ be the virtual edge of $\mathrm{skel}(\mu_i)$ corresponding to the tree edge $\{\mu_i, \mu_{i+1}\}$.

\medskip
\noindent
\textbf{Step 2.}~
Now regard the path $\{\mu_0, \mu_1, \dots,\mu_k\}$ as the root of $T$.
Thus every node $\nu$ outside this path has a unique parent, namely the neighbor of $\nu$ closer to the path.
Let $a_\nu,b_\nu$ be the poles of $\nu$, and let $G_\nu = \mathrm{pert}(\nu) - e_\nu$ be the subgraph of $G$ represented by the subtree rooted at $\nu$, after deleting the parent virtual edge.

Whenever a cycle through $e_s$ and $e_t$ uses the part represented by $\nu$, it enters that part at one pole and leaves it at the other.
Hence, from the viewpoint of the parent, the only relevant information carried by this subtree is the set
$M_\nu \coloneqq \{\,\psi(P)\mid P \text{ is an } a_\nu\text{--}b_\nu \text{ path in } G_\nu\,\}$.
By Lemma~\ref{two-labels}, in linear time we can compute a representative set $\Sigma_\nu\subseteq M_\nu$ of size $\min\{|M_\nu|,2\}$, together with realizing $a_\nu$--$b_\nu$ paths.

This is sufficient because every later use of $M_\nu$ only adds a fixed value to
an element of $M_\nu$.
Since $\alpha\neq \beta$ implies $\alpha+\gamma\neq \beta+\gamma$ for any $\alpha, \beta, \gamma$ in a group $\Gamma$, it is enough to keep one representative if $|M_\nu|=1$, and any two distinct representatives otherwise.

Compute these summaries for all outside nodes $\nu$ adjacent to the path $\{\mu_0, \mu_1, \dots,\mu_k\}$.
More precisely, when processing such a node $\nu$, apply Lemma~\ref{two-labels} to compute a representative set $\Sigma_\nu$ of the label of $a_\nu$--$b_\nu$ paths in $G_\nu$.
After that, every virtual edge corresponding to a subtree outside the path can be replaced by one or two edges of distinct labels corresponding to the summary of that subtree, and we keep the corresponding realizing paths for later reconstruction.

After that, every off-path virtual edge can be replaced by one or two edges of distinct labels corresponding to the summary of the attached subtree, and we keep the corresponding realizing paths for later reconstruction.

\medskip
\noindent
\textbf{Step 3.}~
For each $i=0,1,\dots,k$, let $H_i$ be the graph obtained from $\mathrm{skel}(\mu_i)$ by replacing every virtual edge other than $\ell_i$ and $r_i$ by the already computed summary edges of the corresponding off-path subtree. 
Then $H_i$ is a $\Gamma$-labeled graph, possibly with parallel edges, with two distinguished edges $\ell_i$ and $r_i$, whose labels are set to be $0$. Let $L_i$ be the set of labels of cycles in $H_i$ through both $\ell_i$ and $r_i$, truncated to at most two representatives together with realizing local cycles. 

Any cycle through $e_s$ and $e_t$ induces, for each $i$, a cycle in $H_i$ through both $\ell_i$ and $r_i$, since each off-path subtree is attached to the rest of the graph only through its two poles and can therefore be replaced by the corresponding summary edge. Conversely, any choice of one such local cycle in each $H_i$ combines along the corresponding interface edges, and after expanding each summary edge by its stored realizing path, reconstructs a cycle through $e_s$ and $e_t$ in the original graph. 

Since all interface edges have label $0$, the label of the resulting cycle is the sum of the chosen local labels. Hence, if all local label sets are singletons, then the set of labels of cycles through $e_s$ and $e_t$ is also a singleton, whereas if some $L_j$ has size at least two, fixing arbitrary choices in all other $L_i$ yields two distinct global labels. Therefore the original instance has at least two possible labels if and only if some $L_i$ has size at least two.

For Q- and P-nodes, the only local cycle through both distinguished edges is $\ell_i\cup r_i$, so the local label set is a singleton. For an S-node, the underlying local cycle through both distinguished edges is the skeleton cycle itself, so its possible labels are obtained by summing the summaries along that cycle. Thus only the case that $\mu_i$ is an R-node requires a separate argument. In that case, $H_i$ is triconnected, possibly with parallel edges, and the local problem is exactly the following:

\problemdef{Two Labels of Cycles Through Two Edges in Triconnected Graphs}
{A triconnected $\Gamma$-labeled graph $(G=(V,E), \psi)$ and two distinct edges $e_1,e_2\in E$.}
{Determine whether there exist two cycles $C_1$ and $C_2$ in $G$ that pass through both $e_1$ and $e_2$ with $\psi(C_1) \neq \psi(C_2)$, and if so, find such cycles.}

\begin{theorem}\label{triconnected-problem}
There exists a deterministic $\mathrm{O}(|V| + |E|)$-time algorithm for \textsc{Two Labels of Cycles Through Two Edges in Triconnected Graphs}.
\end{theorem}

We complete the proof of Theorem~\ref{thm:main} assuming Theorem~\ref{triconnected-problem}, which is proved in Section~\ref{sec:triconnected}.

\begin{proof}[Proof of Theorem~\ref{thm:main}]
As observed in Section~\ref{sec:preliminaries}, it suffices to solve \textsc{Two Labels of Cycles Through Two Disjoint Edges} with the aid of the block-cut tree, which can be computed in $\mathrm{O}(|V|+|E|)$ time \cite{hopcroft1973algorithm}.
The discussion in Section~\ref{sec:reduction} so far reduces this problem to \textsc{Two Labels of Cycles Through Two Edges in Triconnected Graphs}, which is solved by Theorem~\ref{triconnected-problem}.
Since the SPQR-tree can be constructed in linear time \cite{gutwenger2000linear} and the total size of all skeletons is $\mathrm{O}(|V|+|E|)$, the whole procedure runs in $\mathrm{O}(|V|+|E|)$ time in total.
The stored realizing summary paths and the realizing local cycles allow us to reconstruct two witness cycles in the original graph within the same linear-time bound.
\end{proof}

\section{Proof of Theorem~\ref{triconnected-problem}}\label{sec:triconnected}

\subsection{Further Reduction}
We first consider an easy case, where $e_1$ and $e_2$ share an end vertex.

\begin{lemma}\label{lem:shared-endpoint}
If $e_1$ and $e_2$ share an end vertex, then the problem can be solved in $\mathrm{O}(|V|+|E|)$ time.
\end{lemma}

\begin{proof}
Let $e_1=x_1y$ and $e_2=x_2y$.
If $x_1 = x_2$, then the two edges are parallel and form a unique cycle passing through both edges, so the problem is trivial.
Suppose that $x_1\neq x_2$, and set $G_0 \coloneqq G-\{y\}$.
Then, a cycle $C$ passing through both $e_1$ and $e_2$ in $G$ corresponds to an $x_1$--$x_2$ path $P$ in $G_0$, where $\psi(C) = \psi(P) + \psi(e_1) + \psi(e_2)$. 
Thus, by applying Lemma~\ref{two-labels} with $G = G_0$, $s = x_1$, and $t = x_2$, the problem is solved in $\mathrm{O}(|V| + |E|)$ time.
\end{proof}

In what follows, we assume that $e_1$ and $e_2$ are disjoint.
Specifically, let $e_1 = x_1y_1$ and $e_2 = x_2y_2$, where $x_1, y_1, x_2, y_2$ are distinct.
Let $G'$ be the graph obtained from $G$ by removing the edges $x_1y_1$ and $x_2y_2$ (if there are parallel edges, remove all of them).
Then, a solution to the decision problem (deciding whether there exist two cycles passing through both $e_1$ and $e_2$ with distinct labels) is as follows (cf.~Theorem~\ref{thm:DN2023}):
\begin{align}\label{eq:decision}
\text{The answer is yes if and only if $G'$ has a non-zero cycle of length at least three.} \tag{$\star$}
\end{align}
We prove \eqref{eq:decision} with an $\mathrm{O}(|V| + |E|)$-time algorithm for construction in the yes case, which completes the proof of Theorem~\ref{triconnected-problem}.

First, we confirm that the answer is no if $G'$ has no non-zero cycle of length at least three.
If $G'$ has no non-zero cycle (including that of length two), then Lemma~\ref{lem:balanced} concludes that $G$ can be shifted so that the $G'$ part has label all $0$, so the answer is clearly no.
Suppose that $G'$ contains parallel edges with distinct labels (forming a non-zero cycle of length two) but not a non-zero cycle of length at least three.
Then, those parallel edges form a block of $G'$ (otherwise, there is another vertex $s$ in the block containing those parallel edges, say between $u$ and $v$, and the block contains a $s$--$\{u, v\}$ $2$-fan, which yield a non-zero cycle intersecting all those three vertices).
In other words, if we merge the parallel edges into a single edge, it becomes a bridge.
Since $G$ is triconnected and $G'$ is obtained just by removing edges between two pairs $\{x_1, y_1\}$ and $\{x_2, y_2\}$, one side of the bridge contains exactly one of $x_1$ and $y_1$ and exactly one of $x_2$ and $y_2$, and the other contains the remaining.
Thus, any pair of vertex-disjoint paths in $G'$ between $\{x_1, y_1\}$ and $\{x_2, y_2\}$ (yielding a cycle in $G$ passing through both $e_1$ and $e_2$) cannot use any of those parallel edges, so the answer does not change if we remove the parallel edges.
This concludes the proof of the no case.

We show the reverse direction (i.e., if $G'$ contains a non-zero cycle of length at least three, then the answer is yes) in a constructive way.

We call $(x_1$--$y_1,\, x_2$--$y_2)$, $(x_1$--$x_2,\, y_1$--$y_2)$, and $(x_1$--$y_2,\, y_1$--$x_2)$ the \emph{pairings} of $\{x_1,y_1,x_2,y_2\}$.
The last two types are called \emph{crossing}.
A \emph{path pair} is a pair $\cP = \{P_1, P_2\}$ of two vertex-disjoint paths $P_1$ and $P_2$ in $G'$ realizing one of the three pairings, i.e., either
\begin{itemize}
    \item $P_1$ is an $x_1$--$y_1$ path and $P_2$ is an $x_2$--$y_2$ path, 
    \item $P_1$ is an $x_1$--$x_2$ path and $P_2$ is a $y_1$--$y_2$ path, or
    \item $P_1$ is an $x_1$--$y_2$ path and $P_2$ is a $y_1$--$x_2$ path.
\end{itemize}
The label of $\cP$ is defined as $\psi(\cP) \coloneqq \psi(P_1)+\psi(P_2)$.
Then, our goal is rephrased as finding two crossing path pairs with distinct labels.
The next lemma relaxes the problem so that non-crossing path pairs are also eligible.

\begin{lemma}\label{lem:2-path-transform}
Suppose that $G'$ contains a non-crossing path pair $\cP = \{P_1, P_2\}$, i.e., of type $(x_1$--$y_1,\,x_2$--$y_2)$.
Then, one can find in $\mathrm{O}(|V|+|E|)$ time one of the following two things:
\begin{itemize}
\setlength{\leftskip}{3mm}
    \item[(i)] Two crossing path pairs $\cQ_1$ and $\cQ_2$ with $\psi(\cQ_1) \neq \psi(\cQ_2)$.
    \item[(ii)] A crossing path pair $\cQ$ with $\psi(\cQ) = \psi(\cP)$.
\end{itemize}
\end{lemma}

\begin{proof}
Apply Lemma~\ref{extended-menger2} to $G$ with $k=3$, $A=V(P_1)$, and $B=V(P_2)$. Since $P_1$ and $P_2$ are vertex-disjoint and each has at least three vertices, we obtain three vertex-disjoint paths $R_1,R_2,R_3$ such that each $R_i$ meets $P_1$ and $P_2$ only at its ends.
Let the ends of $R_1,R_2,R_3$ on $P_1$ be $u_1,u_2,u_3$ in this order from $x_1$ to $y_1$, and let their ends on $P_2$ be $v_1,v_2,v_3$ in this order from $x_2$ to $y_2$.
Then, by the symmetry of indices of $R_i$ and directions of $P_1$ and $P_2$ (i.e., the names of $x_1, y_1$ and $x_2, y_2$), we may assume one of the following two cases (see Figure~\ref{fig:2-path-transform}):
either $R_i$ joins $u_i$ to $v_i$ for all $i$, or $R_1$ joins $u_1$ to $v_1$, $R_2$ joins $u_2$ to $v_3$, and $R_3$ joins $u_3$ to $v_2$.

\begin{figure}[t]
    \centering
    \def\svgwidth{0.9\columnwidth}
    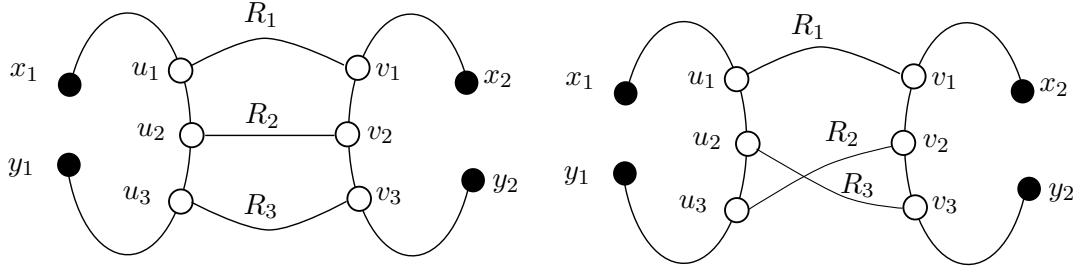 
    \caption{The two possible configurations of the three linking paths $R_1,R_2,R_3$ in Lemma~\ref{lem:2-path-transform}. Up to symmetry, these are the only two relative orders of the end vertices on $P_1$ and $P_2$.}
    \label{fig:2-path-transform}
\end{figure}

First, suppose that $R_i$ joins $u_i$ to $v_i$ for all $i$.
Let $X\coloneqq P_1[x_1,u_1]\cup R_1\cup P_2[v_1,x_2]$, let $Y \coloneqq P_1[y_1,u_2]\cup R_2\cup P_2[v_2,y_2]$, and let $Y' \coloneqq P_1[y_1,u_3]\cup R_3\cup P_2[v_3,y_2]$.
Then, $\{X,Y\}$ and $\{X,Y'\}$ are crossing path pairs.
If $\psi(X)+\psi(Y)\neq\psi(X)+\psi(Y')$, then the first outcome holds. 
Otherwise, $\psi(Y) = \psi(Y')$.
Let $C\coloneqq P_1[u_2,u_3]\cup R_2\cup P_2[v_2,v_3]\cup R_3$. Then $\psi(C)=\psi(Y)+\psi(Y')=0$.
Let $X_0\coloneqq P_1[x_1,u_2]\cup R_2\cup P_2[v_2,x_2]$ and $Y_0 \coloneqq P_1[y_1,u_3]\cup R_3\cup P_2[v_3,y_2]$.
Then, $\{X_0,Y_0\}$ is a crossing path pair again.
By considering the parity of occurrences of each edge in $P_1 \cup P_2 \cup X_0 \cup Y_0$ and $C$, we see $\psi(P_1)+\psi(P_2)+\psi(X_0)+\psi(Y_0)=\psi(C)=0$.
Thus, $\psi(X_0)+\psi(Y_0)=\psi(P_1)+\psi(P_2)$, and the second outcome holds.

Now suppose that $R_1$ joins $u_1$ to $v_1$, $R_2$ joins $u_2$ to $v_3$, and $R_3$ joins $u_3$ to $v_2$. Let $X \coloneqq P_1[x_1,u_1]\cup R_1\cup P_2[v_1,x_2]$, let $Y \coloneqq P_1[y_1,u_2]\cup R_2\cup P_2[v_3,y_2]$, and let $Y'\coloneqq P_1[y_1,u_3]\cup R_3\cup P_2[v_2,y_2]$.
Again, $\{X,Y\}$ and $\{X,Y'\}$ are crossing path pairs.
If $\psi(X)+\psi(Y)\neq\psi(X)+\psi(Y')$, then the first outcome holds.
Otherwise, $\psi(Y)=\psi(Y')$.
Let $C\coloneqq P_1[u_2,u_3]\cup R_3\cup P_2[v_2,v_3]\cup R_2$.
Then, again $\psi(C)=\psi(Y)+\psi(Y')=0$.
Let $X_0 \coloneqq P_1[x_1,u_2]\cup R_2\cup P_2[v_3,y_2]$ and let $Y_0 \coloneqq P_1[y_1,u_3]\cup R_3\cup P_2[v_2,x_2]$.
Then $\{X_0,Y_0\}$ is a crossing path pair.
By considering the parity again, we see $\psi(P_1)+\psi(P_2)+\psi(X_0)+\psi(Y_0)=\psi(C)=0$.
Thus, $\psi(X_0)+\psi(Y_0)=\psi(P_1)+\psi(P_2)$, and the second outcome holds.

The only nontrivial step is the single application of Lemma~\ref{extended-menger2}, and all remaining operations are linear.
Therefore, the whole procedure runs in $\mathrm{O}(|V|+|E|)$ time.
\end{proof}

Lemma~\ref{lem:2-path-transform} reduces the problem of finding two crossing path pairs with distinct labels to the problem of finding two (possibly non-crossing) path pairs with distinct labels.
Indeed, any non-crossing path pair $\cP$ can be transformed into two crossing path pairs of distinct labels or a crossing path pair $\cQ$ with $\psi(\cQ) = \psi(\cP)$.
Applying this replacement to every output non-crossing path pair, we either finish immediately or obtain two crossing path pairs with the same two labels as before.
The remaining task is to prove the following lemma.

\begin{lemma}\label{lem:kernel}
If $G'$ contains a non-zero cycle $C$ of length at least three, then $G'$ contains two path pairs $\cP_1$ and $\cP_2$ with $\psi(\cP_1) \neq \psi(\cP_2)$, which can be found in $\mathrm{O}(|V| + |E|)$ time.
\end{lemma}

\subsection{Constructing Two Path Pairs (Proof of Lemma~\ref{lem:kernel})}
Suppose that $G'$ contains a non-zero cycle of length at least three, and fix such a cycle $C$.
Note that it is easy to find $C$ in linear time by computing the block-cut tree of $G'$ and applying Lemma~\ref{lem:tree-gauge-normalization} to each block of size at least three.

We first show several lemmas, which immediately resolve easy cases in linear time.

\begin{lemma}\label{lem:chord}
Suppose that $G$ contains a $V(C)$--$V(C)$ path $Q$ passing through both $e_1$ and $e_2$ whose inner vertices are disjoint from $V(C)$.
Then there exist two distinct values $\lambda_1,\lambda_2\in\Gamma$, each realized by a path pair in $G'$.
\end{lemma}

\begin{proof}
Suppose that $Q$ starts at $z_1 \in V(C)$, intersects $x_1, y_1, x_2, y_2$ in this order, and ends at $z_2 \in V(C)$.
Then, $Q$ plays a role of chord of the non-zero cycle $C$.
Let $P_1, P_2$ be the two $z_2$--$z_1$ paths on $C$.
We then construct two path pairs $\cP_1 = \{Q[y_1, x_2], Q[y_2, z_2] \cup P_1 \cup Q[z_1, x_1]\}$ and $\cP_2 = \{Q[y_1, x_2], Q[y_2, z_2] \cup P_2 \cup Q[z_1, x_1]\}$, each realizing the same terminal pairing $\{x_1, y_2\}, \{y_1, x_2\}$.
We see
\[\psi(\cP_1) + \psi(\cP_2) = \psi(C) + 2\psi(Q - \{e_1, e_2\}) = \psi(C) \neq 0.\]
Thus, we obtain $\lambda_1 \coloneqq \psi(\cP_1) \neq \psi(\cP_2) \eqqcolon \lambda_2$.
\end{proof}

\begin{lemma}\label{lem:4disjoint}
Suppose that $G'$ contains pairwise vertex-disjoint paths $Q_{x_1},Q_{y_1},Q_{x_2},Q_{y_2}$, each possibly trivial, such that $Q_z$ joins $z$ to a vertex of $C$ for each $z\in\{x_1,y_1,x_2,y_2\}$, and in particular the four end vertices on $C$ are distinct (see Figure~\ref{fig:3distinct-plus}). Then there exist two distinct values $\lambda_1,\lambda_2\in\Gamma$, each realized by a path pair in $G'$.
\end{lemma}

\begin{proof}
For each $z\in\{x_1,y_1,x_2,y_2\}$, let $p_z$ be the first vertex of $C$ on $Q_z$. By Lemma~\ref{first-hit}, we may replace $Q_z$ by its initial segment from $z$ to $p_z$. Then each $Q_z$ meets $C$ only at $p_z$, and the four paths remain pairwise vertex-disjoint.

Let $z_1,z_2,z_3,z_4$ be the vertices $x_1,y_1,x_2,y_2$ in the cyclic order of $p_{z_1},p_{z_2},p_{z_3},p_{z_4}$ on $C$. Let $A_i$ be the $p_{z_i}$--$p_{z_{i+1}}$ path on $C$ whose inner vertices avoid the other two attachment points, where the indices are read modulo $4$. Set $P_1\coloneqq Q_{z_1}\cup A_1\cup Q_{z_2}$, $P_2\coloneqq Q_{z_3}\cup A_3\cup Q_{z_4}$, $S_1\coloneqq Q_{z_2}\cup A_2\cup Q_{z_3}$, and $S_2 \coloneqq Q_{z_4}\cup A_4\cup Q_{z_1}$.
Then $\cP = \{P_1,P_2\}$ and $\cS = \{S_1,S_2\}$ are two path pairs in $G'$, which realize terminal pairings $\{z_1,z_2\},\{z_3,z_4\}$ and $\{z_2,z_3\},\{z_4,z_1\}$, respectively. 
Since $A_1\cup A_2\cup A_3\cup A_4=C$ and each $Q_{z_i}$ appears exactly twice, we have $\psi(\cP)+\psi(\cS)=\psi(C) + 2\sum_{i=1}^4 \psi(Q_{z_i}) = \psi(C) \neq 0$.
Thus, we obtain $\lambda_1 \coloneqq \psi(\cP) \neq \psi(\cS) \eqqcolon \lambda_2$.
\end{proof}

\begin{lemma}\label{lem:3distinct-plus}
Suppose that $y_2 \not\in V(C)$ and $G'$ contains pairwise vertex-disjoint paths $Q_{x_1},Q_{y_1},Q_{x_2}$, each possibly trivial, from $x_1,y_1,x_2$ to distinct vertices $p_{x_1},p_{y_1},p_{x_2}$ of $C$. Let $H\coloneqq C\cup Q_{x_1}\cup Q_{y_1}\cup Q_{x_2}$. Suppose also that $G'$ contains a path $Q_{y_2}$, possibly trivial, from $y_2$ to a vertex of $V(H)\setminus\{p_{x_1},p_{y_1},p_{x_2}\}$ whose inner vertices are disjoint from $H$ (see Figure~\ref{fig:3distinct-plus}). Then there exist two distinct values $\lambda_1,\lambda_2\in\Gamma$, each realized by a path pair in $G'$.
\end{lemma}

\begin{figure}[t]
    \centering
    \scalebox{1.0}{
        \def\svgwidth{0.8\columnwidth}
        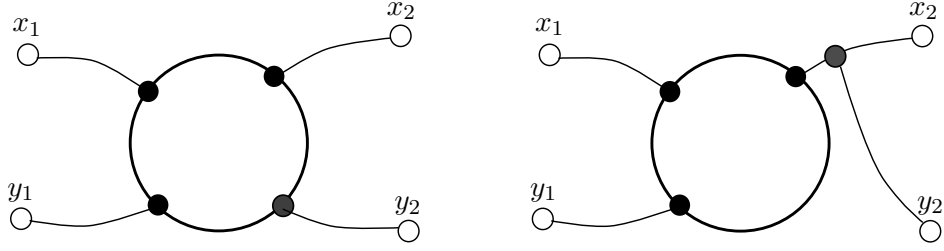
    }
    \caption{Illustration for Lemmas~\ref{lem:4disjoint} and \ref{lem:3distinct-plus}. Left: The situation of Lemma~\ref{lem:4disjoint}, and also the situation of Lemma~\ref{lem:3distinct-plus} when $Q_{y_2}$ ends in $V(C) \setminus \{p_{x_1}, p_{x_2}, p_{y_1}\}$; there are two path pairs realizing distinct labels. Right: The situation of Lemma~\ref{lem:3distinct-plus} when $Q_{y_2}$ ends in $V(H) \setminus V(C)$, i.e., on one of the three paths, say $Q_{x_2}$ here; the two $p_{x_1}$--$p_{y_1}$ paths on $C$ yield two path pairs with distinct labels.} 
    \label{fig:3distinct-plus}
\end{figure}

\begin{proof}
By Lemma~\ref{first-hit}, we may shorten $Q_{x_1},Q_{y_1},Q_{x_2}$ so that each meets $C$ only at its end vertex. Let $p_{y_2}$ be the first vertex of $H$ on $Q_{y_2}$. Again by Lemma~\ref{first-hit}, we may shorten $Q_{y_2}$ so that it meets $H$ only at $p_{y_2}$.

There are two possible positions of $p_{y_2}$ (see Figure~\ref{fig:3distinct-plus}): on the cycle $C$, and on one of the three paths $Q_{x_1},Q_{y_1},Q_{x_2}$.
If $p_{y_2}\in V(C)$, then $Q_{x_1},Q_{y_1},Q_{x_2},Q_{y_2}$ are four pairwise vertex-disjoint paths from $x_1,y_1,x_2,y_2$ to four distinct vertices of $C$, so Lemma~\ref{lem:4disjoint} applies.
Thus we may assume that $p_{y_2}$ lies on one of $Q_{x_1},Q_{y_1},Q_{x_2}$. 

Let $a\in\{x_1,y_1,x_2\}$ be such that $p_{y_2}\in V(Q_a)$, and let $\{b,c\}=\{x_1,y_1,x_2\}\setminus\{a\}$. Let $R, R'$ be the two $p_b$--$p_c$ paths on $C$.
Set $P\coloneqq Q_a[a, p_{y_2}] \cup Q_{y_2}$, $Q \coloneqq Q_b\cup R\cup Q_c$, and $Q'\coloneqq Q_b\cup R'\cup Q_c$. Then $\cP = \{P,Q\}$ and $\cP' = \{P,Q'\}$ are two path pairs in $G'$, which realize the same terminal pairing $\{a,y_2\},\{b,c\}$.
Since $R\cup R'=C$ and the common parts $Q_b$ and $Q_c$ cancel, we have $\psi(Q)+\psi(Q')=\psi(C)\neq 0$. Hence $\psi(Q)\neq\psi(Q')$, and therefore $\psi(\cP)\neq\psi(\cP')$. Thus, we obtain $\lambda_1 \coloneqq \psi(\cP) \neq \psi(\cP') \eqqcolon \lambda_2$.
\end{proof}

We now complete the proof of Lemma~\ref{lem:kernel} based on the following case analysis:
\begin{description}
\setlength{\itemsep}{0mm}
\item[Case 1.] All four vertices $x_1, y_1, x_2, y_2$ are on $C$.
\item[Case 2.] Exactly three among $x_1, y_1, x_2, y_2$ are on $C$.
\item[Case 3.] $V(C) \cap \{x_1, y_1, x_2, y_2\} = \{x_1, y_1\}$ or $\{x_2, y_2\}$.
\item[Case 4.] $|V(C) \cap \{x_1, y_1\}| = 1$ and $|V(C) \cap \{x_2, y_2\}| = 1$.
\item[Case 5.] Exactly one among $x_1, y_1, x_2, y_2$ is on $C$.
\item[Case 6.] None of $x_1, y_1, x_2, y_2$ is on $C$.
\end{description}
It is clear that all cases are indeed completed in linear time from the description.
At a high level, the proof structure is as follows:
Case 6 reduces to one of Cases 1--5, Case 2 reduces to one of Cases 3--5, Case 3 reduces to one of Cases 4--5, Case 5 reduces to Case 4, and Cases 1 and 4 are proved as the terminal cases.

\subsubsection{Case 1: When all four vertices $x_1, y_1, x_2, y_2$ are on $C$}
In this case, Lemma~\ref{lem:4disjoint} immediately completes the proof.

\subsubsection{Case 2: When exactly three among $x_1, y_1, x_2, y_2$ are on $C$}
As we shall see, this case reduces to Case 3, 4, or 5.
By symmetry, we may assume that $x_1, y_1, x_2$ are on $C$.
Let $P_1$ and $P_2$ be the two $x_1$--$y_1$ paths on $C$ such that $P_1$ is disjoint from $x_2$ and $P_2$ intersects $x_2$.
Then, $P_2$ has an inner vertex $x_2$, and $P_1$ also has an inner vertex as $G'$ has no edge between $x_1$ and $y_1$.
Since $G$ is triconnected, $G - \{x_1, y_1\}$ is connected, so it contains a path $Q$ from $V(P_1) \setminus \{x_1, y_1\}$ to $V(P_2) \setminus \{x_1, y_1\}$.
By Lemma~\ref{first-hit}, we may assume that all inner vertices of $Q$ are disjoint from $C$.

\begin{figure}[t]
    \centering
    \scalebox{1.0}{
        \def\svgwidth{0.5\columnwidth}
        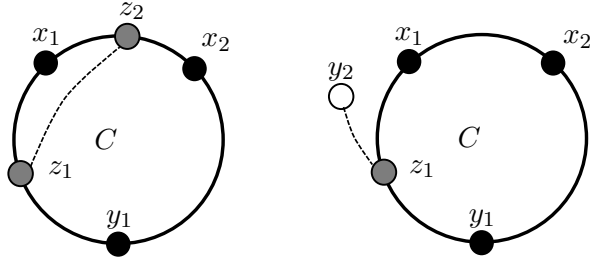
    }
    \caption{Illustration for Case 2. Left: The dashed path is $Q$, which splits the non-zero cycle $C$ into two cycles at least one of which is non-zero. Right: The dashed path is a subpath of $Q$, which enables us to apply Lemma~\ref{lem:4disjoint}.}
    \label{fig:case2}
\end{figure}

Suppose that $Q$ starts at $z_1 \in V(P_1) \setminus \{x_1, y_1\}$ and ends at $z_2 \in V(P_2) \setminus \{x_1, y_1\}$; see Figure~\ref{fig:case2}.
If $Q$ intersects $y_2$, then Lemma~\ref{lem:4disjoint} completes the proof.
Otherwise, $Q$ plays a role of chord of the non-zero cycle $C$, where any of the two $z_1$--$z_2$ paths on $C$ intersects one or two of $x_1, y_1, x_2$. 
Thus, by splitting $C$ with this chord, one can obtain a non-zero cycle intersecting one or two among $x_1, y_1, x_2, y_2$, reducing this case to Case 3, 4, or 5.

\subsubsection{Case 3: When $V(C) \cap \{x_1, y_1, x_2, y_2\} = \{x_1, y_1\}$ or $\{x_2, y_2\}$}
Similarly, we can see that this case reduces to Case 4 or 5.
By symmetry, we may assume $V(C) \cap \{x_1, y_1, x_2, y_2\} = \{x_1, y_1\}$.
Let $P_1$ and $P_2$ be the two $x_1$--$y_1$ paths on $C$.
As in Case 2, both $P_1$ and $P_2$ have an inner vertex.
Since $G$ is triconnected, $G - \{x_1, y_1\}$ is connected, so it contains a path $Q$ from $V(P_1) \setminus \{x_1, y_1\}$ to $V(P_2) \setminus \{x_1, y_1\}$.
By Lemma~\ref{first-hit}, we may assume that all inner vertices of $Q$ are disjoint from $C$.

\begin{figure}[t]
    \centering
    \vspace{3mm}
    \scalebox{1.0}{
        \def\svgwidth{0.65\columnwidth}
        \input{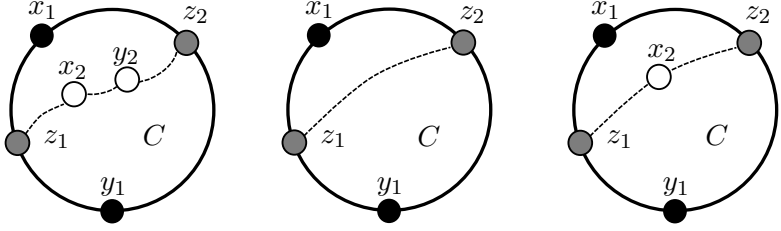}
    }
    \caption{Illustration for Case 3. The dashed path is $Q$. Left: There exist four disjoint paths and Lemma~\ref{lem:4disjoint} applies. Middle, Right: $Q$ splits $C$ into two cycles at least one of which is non-zero.}
    \label{fig:case3}
\end{figure}

Suppose that $Q$ starts at $z_1 \in V(P_1) \setminus \{x_1, y_1\}$ and ends at $z_2 \in V(P_2) \setminus \{x_1, y_1\}$; see Figure~\ref{fig:case3}.
If $Q$ intersects both $x_2$ and $y_2$, then Lemma~\ref{lem:4disjoint} completes the proof.
Otherwise, $Q$ intersects at most one of $x_2$ and $y_2$.
Then, $Q$ plays a role of chord of the non-zero cycle $C$, where any of the two $z_1$--$z_2$ paths on $C$ intersects exactly one of $x_1$ and $y_1$. 
Thus, by splitting $C$ with this chord, one can obtain a non-zero cycle intersecting exactly one of $x_1$ and $y_1$ and at most one of $x_2$ and $y_2$, reducing this case to Case 4 or 5.

\subsubsection{Case 4: When $|V(C) \cap \{x_1, y_1\}| = 1$ and $|V(C) \cap \{x_2, y_2\}| = 1$}
By symmetry, we may assume $V(C) \cap \{x_1, y_1, x_2, y_2\} = \{x_1, x_2\}$.
If there exists a $y_1$--$y_2$ path $Q$ in $G' - V(C)$, then combining with the two $x_1$--$x_2$ paths $P_1$ and $P_2$ on $C$ with distinct labels, one can obtain desired path pairs $\{P_1, Q\}$ and $\{P_2, Q\}$.
Thus, we may assume that any $y_1$--$y_2$ path in $G'$ intersects $V(C)$.
Also, if there exist two vertex-disjoint paths in $G - \{x_1, x_2\} = G' - \{x_1, x_2\}$ between $\{y_1, y_2\}$ and $V(C) \setminus \{x_1, x_2\}$ (which can be found by Lemma~\ref{extended-menger2}), Lemma~\ref{lem:4disjoint} completes the proof. 
Since $G$ is triconnected, we may assume that $G - \{x_1, x_2\}$ has a $1$-cut separating $\{y_1, y_2\}$ and $V(C) \setminus \{x_1, x_2\}$ by Menger's theorem, which is on $C$ (otherwise $G' - V(C)$ has a $y_1$--$y_2$ path through the $1$-cut).

Let $z \in V(C)$ be a vertex such that $G - \{x_1, x_2, z\}$ has no path between $\{y_1, y_2\}$ and $V(C) \setminus \{x_1, x_2, z\}$.
By Lemma~\ref{extended-menger2}, one can find in linear time
\begin{itemize}
    \item a $y_1$--$\{z, x_2\}$ $2$-fan in $G - \{x_1\}$, and
    \item a $y_2$--$\{z, x_1\}$ $2$-fan in $G - \{x_2\}$.
\end{itemize}
Let $Q_1, Q_2, Q_3, Q_4$ be the $y_1$--$z$, $y_1$--$x_2$, $y_2$--$z$, $y_2$--$x_1$ paths in the $2$-fans, respectively.
Since $\{x_1, x_2, z\}$ is a $3$-cut in $G$ separating $\{y_1, y_2\}$ and $V(C) \setminus \{x_1, x_2, z\}$, all of those paths are disjoint from $V(C) \setminus \{x_1, x_2, z\}$ (see Figure~\ref{fig:case4}).
Also, since $G' - V(C)$ has no $y_1$--$y_2$ path, 
\begin{itemize}
    \item $Q_1$ and $Q_2$ do not pass through $e_2$ (otherwise they intersect $y_2$ before reaching $\{z, x_2\}$), so $Q_1$ and $Q_2$ are in $G'$,
    \item $Q_3$ and $Q_4$ do not pass through $e_1$ (otherwise they intersect $y_1$ before reaching $\{z, x_1\}$), so $Q_3$ and $Q_4$ are in $G'$, and
    \item all four paths share only the end vertices (otherwise $G' - V(C)$ has a $y_1$--$y_2$ path through the shared inner vertex).
\end{itemize}
Let $P_1, P_2, P_3$ be the $x_1$--$x_2$, $x_2$--$z$, $z$--$x_1$ paths on $C$ that are disjoint from $z, x_1, x_2$, respectively.
Consider four path pairs $\cP_1 = \{P_1, Q_1 \cup Q_3\}$, $\cP_2 = \{Q_2, Q_4\}$, $\cP_3 = \{Q_2, Q_3 \cup P_3\}$, and $\cP_4 = \{Q_4, Q_1 \cup P_2\}$.
They are indeed path pairs by the above discussion, whose labels satisfy
\[\sum_{i=1}^4\psi(\cP_i) = \sum_{i=1}^3\psi(P_i) + 2\sum_{i=1}^4\psi(Q_i) = \sum_{i=1}^3\psi(P_i) = \psi(C) \neq 0.\]
This concludes that there exists a pair of $\cP_i$ and $\cP_j$ with $\psi(\cP_i) \neq \psi(\cP_j)$ (otherwise $\sum_{i=1}^4\psi(\cP_i) = 4\alpha = 0$ for some $\alpha \in \Gamma$), completing the proof.

\begin{figure}[t]
    \centering
    \scalebox{1.0}{
        \def\svgwidth{0.3\columnwidth}
        %% Creator: Inkscape 1.4.3 (0d15f75, 2025-12-25), www.inkscape.org
%% PDF/EPS/PS + LaTeX output extension by Johan Engelen, 2010
%% Accompanies image file 'Case_4.pdf' (pdf, eps, ps)
%%
%% To include the image in your LaTeX document, write
%%   \input{<filename>.pdf_tex}
%%  instead of
%%   \includegraphics{<filename>.pdf}
%% To scale the image, write
%%   \def\svgwidth{<desired width>}
%%   \input{<filename>.pdf_tex}
%%  instead of
%%   \includegraphics[width=<desired width>]{<filename>.pdf}
%%
%% Images with a different path to the parent latex file can
%% be accessed with the `import' package (which may need to be
%% installed) using
%%   \usepackage{import}
%% in the preamble, and then including the image with
%%   \import{<path to file>}{<filename>.pdf_tex}
%% Alternatively, one can specify
%%   \graphicspath{{<path to file>/}}
%% 
%% For more information, please see info/svg-inkscape on CTAN:
%%   http://tug.ctan.org/tex-archive/info/svg-inkscape
%%
\begingroup%
  \makeatletter%
  \providecommand\color[2][]{%
    \errmessage{(Inkscape) Color is used for the text in Inkscape, but the package 'color.sty' is not loaded}%
    \renewcommand\color[2][]{}%
  }%
  \providecommand\transparent[1]{%
    \errmessage{(Inkscape) Transparency is used (non-zero) for the text in Inkscape, but the package 'transparent.sty' is not loaded}%
    \renewcommand\transparent[1]{}%
  }%
  \providecommand\rotatebox[2]{#2}%
  \newcommand*\fsize{\dimexpr\f@size pt\relax}%
  \newcommand*\lineheight[1]{\fontsize{\fsize}{#1\fsize}\selectfont}%
  \ifx\svgwidth\undefined%
    \setlength{\unitlength}{139.99937403bp}%
    \ifx\svgscale\undefined%
      \relax%
    \else%
      \setlength{\unitlength}{\unitlength * \real{\svgscale}}%
    \fi%
  \else%
    \setlength{\unitlength}{\svgwidth}%
  \fi%
  \global\let\svgwidth\undefined%
  \global\let\svgscale\undefined%
  \makeatother%
  \begin{picture}(1,0.56732945)%
    \lineheight{1}%
    \setlength\tabcolsep{0pt}%
    \put(0,0){\includegraphics[width=\unitlength,page=1]{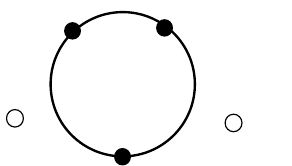}}%
    \put(0.20546161,0.51804902){\color[rgb]{0,0,0}\makebox(0,0)[lt]{\lineheight{1.25}\smash{\begin{tabular}[t]{l}$x_2$\end{tabular}}}}%
    \put(0.7713674,0.06147066){\color[rgb]{0,0,0}\makebox(0,0)[lt]{\lineheight{1.25}\smash{\begin{tabular}[t]{l}$y_2$\end{tabular}}}}%
    \put(0.38953479,0.07833765){\color[rgb]{0,0,0}\makebox(0,0)[lt]{\lineheight{1.25}\smash{\begin{tabular}[t]{l}$z$\end{tabular}}}}%
    \put(0.53804514,0.52089448){\color[rgb]{0,0,0}\makebox(0,0)[lt]{\lineheight{1.25}\smash{\begin{tabular}[t]{l}$x_1$\end{tabular}}}}%
    \put(-0.01609571,0.07915715){\color[rgb]{0,0,0}\makebox(0,0)[lt]{\lineheight{1.25}\smash{\begin{tabular}[t]{l}$y_1$\end{tabular}}}}%
    \put(0.38109712,0.24686763){\color[rgb]{0,0,0}\makebox(0,0)[lt]{\lineheight{1.25}\smash{\begin{tabular}[t]{l}$C$\end{tabular}}}}%
    \put(0,0){\includegraphics[width=\unitlength,page=2]{Case_4.pdf}}%
  \end{picture}%
\endgroup%

    }
    \caption{Illustration for Case 4. The $2$-fans do not share inner vertices with each other or $C$.}
    \label{fig:case4}
\end{figure}

\subsubsection{Case 5: When exactly one among $x_1, y_1, x_2, y_2$ is on $C$}
By symmetry, we may assume $V(C) \cap \{x_1, y_1, x_2, y_2\} = \{x_1\}$.
Since $G$ is triconnected, there exist two vertex-disjoint paths $Q_1, Q_2$ between $(V(C) \setminus \{x_1\}) \cup \{y_1\}$ and $\{x_2, y_2\}$ in $G - \{x_1\}$ such that $Q_1$ starts at $y_1$ and $Q_2$ has no inner vertex on $C$, which can be found by Lemma~\ref{extended-menger2} (apply to the graph obtained from $G - \{x_1\}$ by adding a vertex $v_C$ adjacent to all vertices in $V(C) \setminus \{x_1\}$ with $A = \{v_C, y_1\}$ and $B = \{x_2, y_2\}$, remove $v_C$ from one of the obtained paths, and use Lemma~\ref{first-hit} if necessary); see Figure~\ref{fig:case5_1}.
Note that neither $Q_1$ nor $Q_2$ can use $e_2$, so they are paths in $G'$.
If $Q_1$ is disjoint from $V(C)$, then by concatenating $e_1, Q_1, e_2, Q_2$, we obtain a $V(C)$--$V(C)$ path in $G$ passing through both $e_1$ and $e_2$ whose inner vertices are disjoint from $V(C)$.
Then, Lemma~\ref{lem:chord} completes the proof.
Also, if $Q_1$ contains at least two vertices on $C$, then by letting $z_1$ and $z_2$ be the first and last ones, respectively, and $w$ be the end vertex of $Q_2$ on $C$, we obtain four vertex-disjoint paths between $\{x_1, y_1, x_2, y_2\}$ and $\{x_1, z_1, z_2, w\}$.
Then, Lemma~\ref{lem:4disjoint} completes the proof.

\begin{figure}[t]
    \centering
    \vspace{3mm}
    \scalebox{1.0}{
        \def\svgwidth{0.9\columnwidth}
        \input{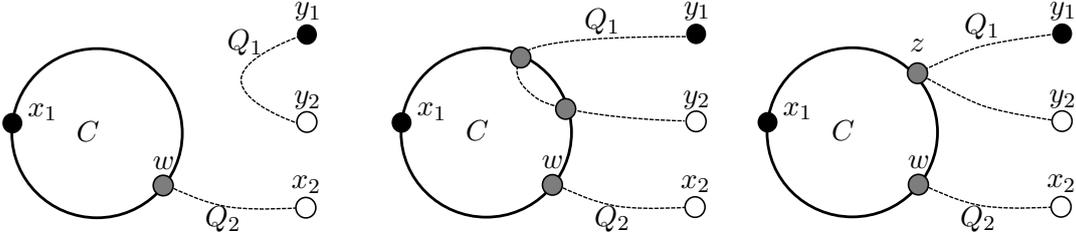}
    }
    \caption{Illustration for Case~5. Left: $Q_1$ is disjoint from $V(C)$, and then Lemma~\ref{lem:chord} applies. Middle: $Q_1$ contains at least two vertices on $C$, and there exist four distinct paths, so Lemma~\ref{lem:4disjoint} applies. Right: $Q_1$ contains exactly one vertex on $C$. We focus on this case below.}
    \label{fig:case5_1}
\end{figure}

Otherwise, $Q_1$ contains exactly one vertex on $C$, say $z$.
Suppose that $Q_2$ is between $w \in V(C) \setminus \{x_1\}$ and $x_2$.
Let $H_1 \coloneqq (C \cup Q_1[z, y_2] \cup Q_2) - \{x_1\}$ and $H_2 \coloneqq (C \cup Q_1[z, y_1] \cup Q_2) - \{x_2\}$.
By Lemma~\ref{extended-menger2}, one can find in linear time
\begin{itemize}
    \item a $y_1$--$V(H_1)$ $2$-fan in $G - \{x_1\}$, and
    \item a $y_2$--$V(H_2)$ $2$-fan in $G - \{x_2\}$.
\end{itemize}
Let $R_1, R_2$ be the $y_1$--$V(H_1)$ paths and $R_3, R_4$ be the $y_2$--$V(H_2)$ paths in the $2$-fans.

If $R_1$ or $R_2$ ends at $V(H_1) \setminus \{z, w\}$ or $R_3$ or $R_4$ ends at $V(H_2) \setminus \{x_1, z, w\}$, then Lemma~\ref{lem:3distinct-plus} completes the proof.
Moreover, if $R_3$ or $R_4$ ends at $x_1$, $R_3 \cup R_4$ plays a role of chord of the non-zero cycle $C$ intersecting both $x_1$ and $y_2$.
By splitting $C$ with this chord, one can obtain a non-zero cycle intersecting $x_1$ and $y_2$ and disjoint from $y_1$ and $x_2$, reducing this case to Case 4.
Thus, we may assume that $R_1$ and $R_3$ end at $z$ and $R_2$ and $R_4$ end at $w$.
Furthermore, if some pair of the four paths share an inner vertex, it contains a $y_1$--$y_2$ path disjoint from $C \cup Q_2$ and Lemma~\ref{lem:chord} applies again.
Thus, we may assume that they are internally disjoint; see Figure~\ref{fig:case5_2}.

Let $H_3 \coloneqq (C \cup R_1 \cup R_2 \cup R_3 \cup R_4) - \{y_2\}$.
By Lemma~\ref{extended-menger2}, one can find in linear time a $x_2$--$V(H_3)$ $2$-fan in $G - \{y_2\}$.
Let $S_1, S_2$ be the $x_2$--$V(H_3)$ paths in the $2$-fan.
If $S_1$ or $S_2$ ends at $V(H_3) \setminus \{x_1, z, w\}$, Lemma~\ref{lem:3distinct-plus} completes the proof.
Moreover, if $S_1$ or $S_2$ ends at $x_1$, $S_1 \cup S_2$ plays a role of chord of $C$ intersecting both $x_1$ and $x_2$.
By splitting $C$ with this chord, one can reduce this case to Case 4 again.
Thus, we may assume that $S_1$ ends at $z$ and $S_2$ ends at $w$.
Note that the six paths do not share any inner vertices by definition.

\begin{figure}[t]
    \centering
    \scalebox{1.0}{
        \def\svgwidth{0.9\columnwidth}
        \input{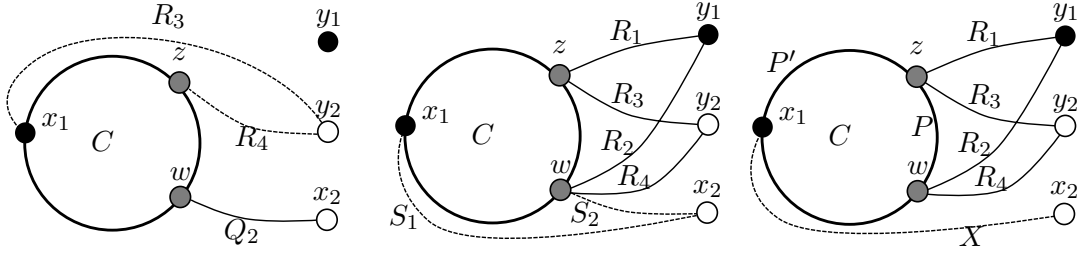}
    }
    \caption{Illustration for Case~5. Left: $R_3 \cup R_4$ splits $C$ into two cycles at least one of which is non-zero. Middle: $S_1 \cup S_2$ splits $C$ into two cycles at least one of which is non-zero. Right: An $x_2$--$x_1$ path $X$ together with $R_1$, $R_3$, $R_4$, and $P$ yields the conclusion.}
    \label{fig:case5_2}
\end{figure}

Finally, since $G$ is triconnected, $G - \{z, w\}$ contains an $\{x_2, y_2\}$--$\{x_1, y_1\}$ path $X$.
By Lemma~\ref{first-hit}, we may assume that $X$ passes through neither $e_1$ nor $e_2$, and hence is in $G'$.
By symmetry, we may assume that $X$ starts at $x_2$.
If the first vertex of $(C \cup R_1 \cup R_2 \cup R_3 \cup R_4) - \{z, w\}$ on $X$ is not $x_1$, Lemma~\ref{lem:3distinct-plus} completes the proof.
Otherwise, $X$ is an $x_2$--$x_1$ path disjoint from $(C \cup R_1 \cup R_2 \cup R_3 \cup R_4) - \{x_1, z, w\}$; see Figure~\ref{fig:case5_2} again.
Let $P$ be the $z$--$w$ path on $C$ disjoint from $x_1$.
If the two $y_1$--$y_2$ paths $Y_1 = R_1 \cup R_3$ and $Y_2 = R_1 \cup P \cup R_4$ have distinct labels, we obtain desired path pairs $\{X, Y_1\}$ and $\{X, Y_2\}$.
Otherwise, $0 = \psi(Y_1) + \psi(Y_2) = \psi(R_3) + \psi(P) + \psi(R_4)$.
Thus, letting $P'$ be the $z$--$w$ path on $C$ intersecting $x_1$, we obtain a non-zero cycle $C' = R_3 \cup P' \cup R_4$ intersecting $x_1$ and $y_2$ but not $y_1$ and $x_2$, reducing this case to Case 4.

\subsubsection{Case 6: When none of $x_1, y_1, x_2, y_2$ is on $C$}
This case reduces to one of Cases 1--5.
Take an $x_1$--$V(C)$ $3$-fan in $G$ by Lemma~\ref{extended-menger2}.
If both $e_1$ and $e_2$ are used in the $3$-fan, two paths among them form a $V(C)$--$V(C)$ path in $G$ passing through both $e_1$ and $e_2$ whose inner vertices are disjoint from $V(C)$.
Then, Lemma~\ref{lem:chord} completes the proof.
Otherwise, an $x_1$--$V(C)$ $2$-fan in $G'$ is obtained from the $3$-fan.
This forms a $V(C)$--$V(C)$ path $Q$ intersecting $x_1$ whose inner vertices are disjoint from $V(C)$.
Then, by splitting $C$ with $Q$, one can obtain a non-zero cycle intersecting $x_1$, reducing this case to one of Cases 1--5; see Figure~\ref{fig:case6}.

\begin{figure}[t]
    \centering
    \vspace{3mm}
    \scalebox{1.0}{
        \def\svgwidth{0.8\columnwidth}
        \input{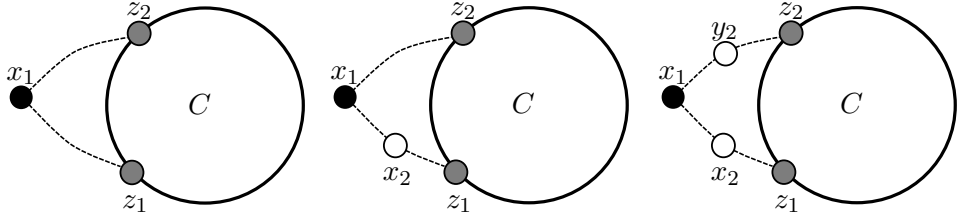}
    }
    \caption{Illustration for Case 6. The dashed path $Q$ splits $C$ into two cycles intersecting $x_1$ at least one of which is non-zero.}
    \label{fig:case6}
\end{figure}

\bigskip
Thus, we complete the proof of Lemma~\ref{lem:kernel}.
Recall that by Lemmas~\ref{lem:shared-endpoint} and \ref{lem:2-path-transform}, our remaining task for proving Theorem~\ref{triconnected-problem} is to find two path pairs of distinct labels in linear time.
Lemma~\ref{lem:kernel} gives a solution to this with a complete characterization \eqref{eq:decision} of the yes instances (which extends Theorem~\ref{thm:DN2023}).
Combining with the discussion just after stating Theorem~\ref{triconnected-problem}, we conclude the proof of Theorem~\ref{thm:main}.

\section*{Acknowledgments}
We thank the anonymous reviewers of the earlier version for their helpful comments.

This work was conducted in part through Experts in Information Science Program in National Institute of Informatics, implemented under JST Science and Technology Challenge Program for Next Generation (STELLA).
This work was also supported by JSPS KAKENHI Grant Number JP25H01114, JST CRONOS Japan Grant Number JPMJCS24K2, and JST ASPIRE Japan Grant Number JPMJAP2520.

\subsection*{AI Declarations}
The authors used generative AIs to assist in manuscript review and language editing.
All mathematical arguments, references, and suggested revisions were verified by the authors.
The authors are responsible for all descriptions in the submitted manuscript.

\bibliographystyle{plain}
\bibliography{ref}

\end{document}